\documentclass[a4j, 10.5pt]{article}
\usepackage{wakaiki_singlecolumn_Eng}
\newcommand{\diag}{\mathop{\rm diag}\nolimits}

\title{
Quantized Output Feedback Stabilization of Switched Linear Systems
}

\author{Masashi Wakaiki and Yutaka Yamamoto
\thanks{M. Wakaiki and Y. Yamamoto are with the Department of Applied Analysis and Complex
Dynamical Systems, Graduate School of Informatics, Kyoto University, Kyoto
606-8501, Japan
(e-mail:{\tt  \ wakaiki@acs.i.kyoto-u.ac.jp};
{\tt \ yy@i.kyoto-u.ac.jp}).}%
}

\begin{document}

\maketitle
\thispagestyle{empty}
\pagestyle{empty}

\begin{abstract}
This paper studies the problem of stabilizing a continuous-time switched linear system
by quantized output feedback.
We assume that the quantized outputs
and the switching signal are available to the controller at all time.
We develop an encoding strategy by using multiple Lyapunov functions 
and an average dwell time property.
The encoding strategy is based on the results in the case of a single mode, 
and it requires an additional adjustment of the ``zoom'' parameter
at every switching time. 
\end{abstract}

\section{INTRODUCTION}
This paper studies the quantized control problem for switched systems.
For linear time-invariant systems, various approaches to quantized control have
been developed: 
Lyapunov-based methods~\cite{Ishii2002Book,Brockett2000, Liberzon2003Automatica},
optimization with $\ell^{\infty}$ norm~\cite{Asuma2008}, etc.
In contrast, few results of quantized control are generalized to switched systems
in spite of a wide range of their applications.
Recently, based on the results in \cite{Liberzon2003} for a single mode,
Liberzon~\cite{Liberzon2013} has developed an encoding
and control strategy achieving the global asymptotic stability of 
sampled-data switched systems with quantized state feedback.
Also, quantized state feedback stabilization is discussed 
for discrete-time Markov jump linear systems in \cite{Nair2003, Zhang2009, Ling2010}.
However, stabilization of switched systems by quantized {\em output}
feedback has not yet explored.

Here we consider a continuous-time switching linear system, whose
quantized outputs and switching signal are transmitted to the controller at all times.
The objective of this paper is to extend the encoding method 
of \cite{Brockett2000, Liberzon2003Automatica} for non-switched systems.
The key point of the earlier studies is that certain level sets of a Lyapunov function
are invariant regions.
The difficulty of switched systems is that such level sets change at every switching time.
Therefore, at the ``zooming-in'' stage,
non-switched systems require only
periodic reduction of the ``zoom'' parameter of quantizers,
whereas in switched systems, 
we need to adjust the parameter after each switch.
We assume that the average dwell time~\cite{Hespanha1999CDC} 
of the switching signal is
large enough, and develop an output encoding for
global asymptotic stabilization by using multiple Lyapunov functions.

This paper is organized as follows. In Section II, we explain the components 
of the closed-loop system one by one 
and then give the main result, Theorem \ref{thm:stability_theorem}.
Section III is devoted to its proof.
We present a numerical example in Section IV and finally conclude this paper 
in Section V.


{\em Notation:~}
Let $\lambda_{\min}(P)$ and $\lambda_{\max}(P)$ denote 
the smallest and the largest eigenvalue of $P \in \mathbb{R}^{\sf n\times n}$.
Let $M^{\top}$ denote the transpose of $M \in \mathbb{R}^{\sf m\times n}$.

The Euclidean norm of $v \in \mathbb{R}^{\sf n}$ is
denoted by $|v| = (v^*v)^{1/2}$. 
The Euclidean induced norm of $M \in \mathbb{R}^{\sf m\times n}$ is defined by
$\|M\| = \sup \{  |Mv |:~v\in \mathbb{R}^{\sf n},~|v|= 1 \}$,
which equals the largest singular value of $M$.

For a piecewise continuous function $f:~\mathbb{R} \to \mathbb{R}$,
its left-sided limit at $t_0 \in \mathbb{R}$ is denoted by $\lim_{t \nearrow t_0}f(t)$.

\section{QUANTIZED OUTPUT FEEDBACK STABILIZATION FOR \\ SWITCHED SYSTEMS}
In this section, we first define switched systems and 
construct quantizers and controllers
based on the non-switched case in \cite{Liberzon2003Automatica}.
Next we present the main result, Theorem \ref{thm:stability_theorem}.
This theorem guarantees 
the existence of a quantizer leading to the globally asymptotic stability 
of the closed-loop system under an average dwell time assumption.

\subsection{Switched linear systems}
Consider the switched linear system
\begin{equation}
\label{eq:SLS}
\dot x = A_{\sigma}x+B_{\sigma}u,\quad y = C_{\sigma}x,
\end{equation}
where $x(t) \in \mathbb{R}^{\sf{n}}$ is the state,
$u(t) \in \mathbb{R}^{\sf{m}}$ is the control input, and
$y(t) \in \mathbb{R}^{\sf{p}}$ is the output.
For a finite index set $\mathcal{P}$,
$\sigma:[0,\infty) \to \mathcal{P}$ is right-continuous and
piecewise constant.
We call $\sigma$ {\em switching signal} and
the discontinuities of $\sigma$ {\em switching times}.
Let denote by $N_{\sigma}(t,s)$
the number of discontinuities of $\sigma$ on the interval $(s,t]$. 

Assumptions on the switched system \eqref{eq:SLS} are as follows.
\begin{assumption}
{\em
\label{ass:system}
For every $p \in \mathcal{P}$, $(A_p, B_p)$ is stabilizable and 
$(C_p, A_p)$ is observable. 
We choose $K_p \in \mathbb{R}^{\sf m \times n}$ and 
$L_p \in \mathbb{R}^{\sf n \times p}$ so that
$A_p+B_pK_p$ and $A_p+L_pC_p$ are Hurwitz.

Furthermore, the switching signal $\sigma$ has an average dwell time~\cite{Hespanha1999CDC},
i.e., there exist $\tau_a>0$ and $N_0 \geq 1$ such that
\begin{equation}
\label{eq:ADT_cond}
N_{\sigma}(t,s) \leq N_0 + \frac{t-s}{\tau_a}
\qquad (t > s \geq 0).
\end{equation}
}
\end{assumption}
\subsection{Quantizer}
In this paper, 
we use the following class of quantizers proposed in \cite{Liberzon2003Automatica}.

Let $\mathcal{Q}$ be a finite subset of $\mathbb{R}^{\sf{p}}$.
A quantizer is a piecewise constant function 
$q:\mathbb{R}^{\sf{p}} \to \mathcal{Q}$.
This geometrically implies that
$\mathbb{R^{\sf{p}}}$ is divided into the quantized regions
$\{y \in \mathbb{R^{\sf{p}}}:~q(y) = y_i \}$ $(y_i \in \mathcal{Q})$.
For the quantizer $q$, there exist positive numbers $M$ and $\Delta$ with
$M > \Delta$ such that
\begin{align}
\label{eq:quantizer_cond1_nonSaturation}
|y| \leq M &\quad \Rightarrow \quad |q(y) - y| \leq \Delta \\
\label{eq:quantizer_cond2}
|y| > M &\quad \Rightarrow \quad |q(y) | > M - \Delta.
\end{align}
The former condition \eqref{eq:quantizer_cond1_nonSaturation} gives
an upper bound of the quantization error when the quantizer is not saturated.
The latter \eqref{eq:quantizer_cond2} is used for the detection of the saturation.

We make the following assumption on the behavior of the quantizer $q$ near the origin:
\begin{assumption}[\cite{Liberzon2003Automatica,Liberzon2007}]
\label{ass:near origin}
{\em 
There exists $\Delta_0 > 0$ such that $q(y)=0$
for every $y \in \mathbb{R}^{\sf{p}}$ with $|y|\leq \Delta_0$.
}
\end{assumption}
This assumption is necessary for the Lyapunov stability of the closed-loop system. 

We give the above quantizers the following adjustable parameter $\mu > 0$:
\begin{equation}
\label{eq:zoom_q}
q_{\mu}(y) = \mu q\left( \frac{y}{\mu}\right).
\end{equation}
In \eqref{eq:zoom_q}, $\mu$ is regarded as a ``zoom'' variable,
and $q_{\mu(t)}(y(t))$ is the data on $y(t)$ transmitted to the controller.
We need to change $\mu$ to obtain accurate information of $y$. 
The reader can refer to 
\cite{Liberzon2003Automatica, Liberzon2007, Liberzon2003Book}
for further discussion.

\begin{remark}
The quantized output $q_{\mu}(y)$ may chatter on the boundaries
among the quantization regions. Hence if we generate $u$ by $q_{\mu}(y)$,
the solutions of \eqref{eq:SLS} must be interpreted 
in the sense of Filippov~\cite{Filippov1988,Cortes2008}.
However this generalization does not affect our Lyapunov-based analysis in this work,
because we will work with a single quadratic Lyapunov function between switching times.
See also \cite{Skafidas1999}, which presents
a Lyapunov-based analysis with the generalized solutions for switched controller systems.
\end{remark}

\subsection{Construction of controllers}
We construct the following dynamic output feedback law based on 
the standard Luenberger observers:
\begin{equation}
\label{eq:controller_def}
\dot \xi = (A_{\sigma}+L_{\sigma}C_{\sigma})\xi + 
B_{\sigma}u - L_{\sigma}q_{\mu}(y),\quad
u = K_{\sigma} \xi,
\end{equation}
where $\xi \in \mathbb{R}^{\sf{n}}$ is the estimated state.
Then the closed-loop system is given by
\begin{align}
\label{eq:ClosedSystem}
\begin{array}{l}
\dot x = A_{\sigma}x + B_{\sigma}K_{\sigma}\xi \\
\dot \xi = (A_{\sigma}+L_{\sigma}C_{\sigma})\xi +
B_{\sigma}K_{\sigma}\xi - L_{\sigma} q_{\mu}(y).
\end{array}
\end{align}
If we define $z$ and $F_{\sigma}$ by
\begin{equation*}
z := 
\begin{bmatrix}
x \\ x - \xi
\end{bmatrix},\quad
F_{\sigma} := 
\begin{bmatrix}
A_{\sigma}+B_{\sigma}K_{\sigma} & -B_{\sigma}K_{\sigma} \\
0 & A_{\sigma}+L_{\sigma}C_{\sigma}
\end{bmatrix},
\end{equation*}
then we rewrite \eqref{eq:ClosedSystem} in the form
\begin{equation}
\label{eq:ClosedSystem_simple}
\dot z = F_{\sigma} z +
\begin{bmatrix} 
0 \\ L_{\sigma}
\end{bmatrix}
(q_{\mu}(y) - y) .
\end{equation}
Since $F_p$ is Hurwitz for every $p \in \mathcal{P}$,
there exist positive-definite matrices 
$P_p, Q_p \in \mathbb{R}^{2{\sf n} \rm \times 2 {\sf n}}$ such that
\begin{equation}
\label{eq:Lyapunov_cont}
F_p^{\top}P_p + P_pF_p = -Q_p
\qquad (p \in \mathcal{P}).
\end{equation}
We define 
$\overline \lambda_P$, $\underline \lambda_P$, 
$\underline \lambda_Q$, and $C_{\max}$ by
\begin{align}
\label{eq:alpha_def}
\begin{array}{c}
\overline \lambda_P := 
{\displaystyle \max_{p \in \mathcal{P}} \lambda_{\max}(P_p)},
\quad
\underline \lambda_P  := 
{\displaystyle \min_{p \in \mathcal{P}} \lambda_{\min}(P_p)}, \\[8pt]
\underline \lambda_Q  := 
{\displaystyle \min_{p \in \mathcal{P}} \lambda_{\min}(Q_p)},
\quad
C_{\max} := 
{\displaystyle \max_{p \in \mathcal{P}}\|C_p\|}.
\end{array}
\end{align}

\subsection{Main result}
As in the non-switched case~\cite{Liberzon2003Automatica},
by adjusting the ``zoom'' parameter $\mu$,
we can achieve the global asymptotic stability of 
the closed-loop system \eqref{eq:ClosedSystem_simple}
in Fig.~\ref{fig:CSSQOF}.

\begin{figure}[bt]
 \centering
 \includegraphics[width = 7cm,bb= 20 20 705 500,clip]
{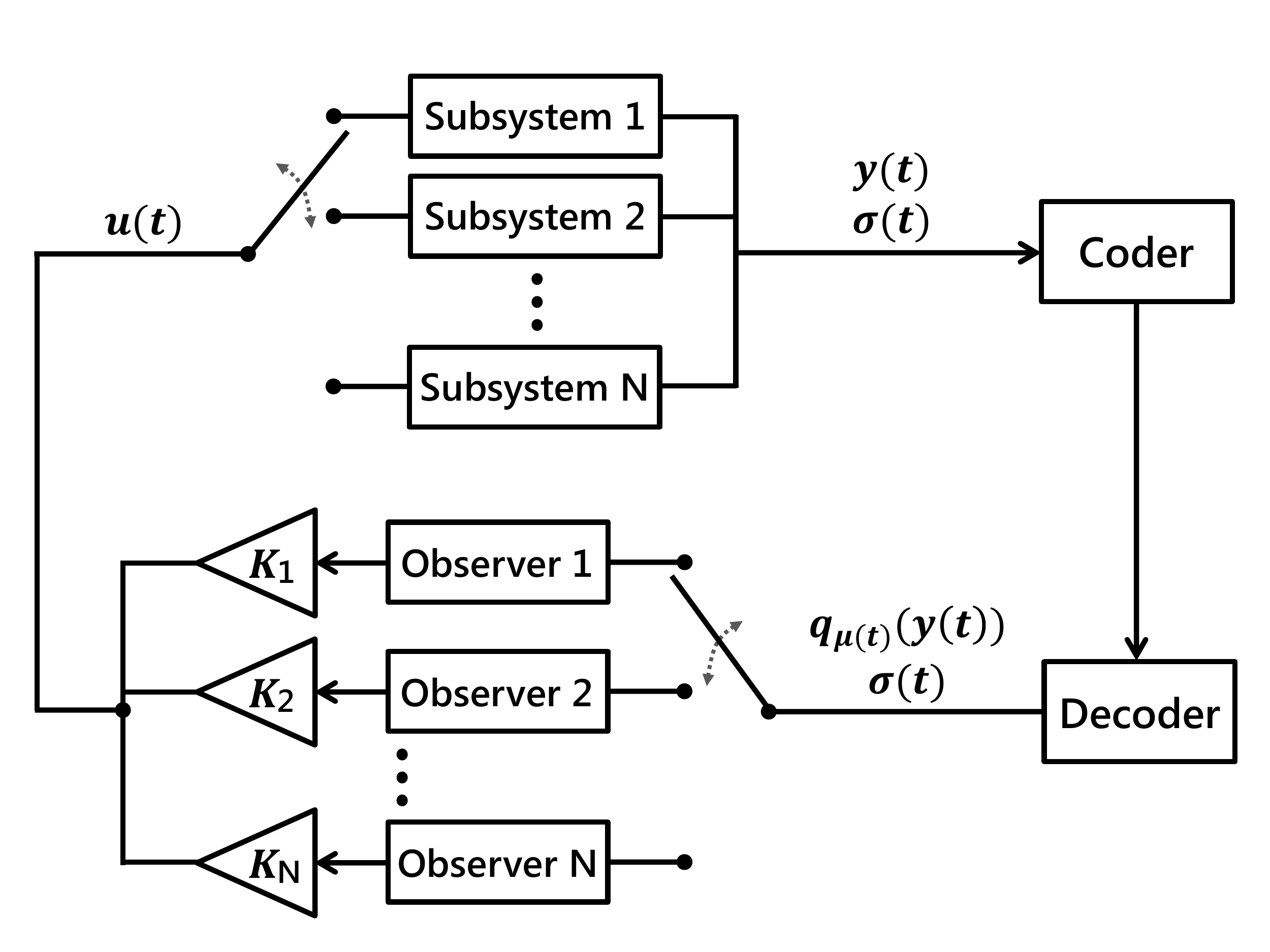}
 \caption{Continuous-time switched system with quantized output feedback.}
 \label{fig:CSSQOF}
 \end{figure}

\begin{theorem}
\label{thm:stability_theorem}
{\em
Define $\Theta$ by
\begin{equation}
\label{eq:Theta_def}
\Theta := \frac{2 \max_{p\in \mathcal{P}} \|P_p\hat L_p\|}{\underline \lambda_Q},
\text{~~where~~}
\hat L_p := 
\begin{bmatrix} 
0 \\ L_{p}
\end{bmatrix}.
\end{equation}
and let $M$ be large enough to satisfy
\begin{equation}
\label{eq:M_Delta_cond_main_thm}
M > \max
\left\{ 2\Delta,\quad
\sqrt{\frac{\overline \lambda_P}{\underline \lambda_P}}
\Theta \Delta C_{\max}
\right\}.
\end{equation}
If the average dwell time $\tau_a$ is longer than a certain value, then
there exists a piecewise constant function $\mu$ such that
the closed-loop system \eqref{eq:ClosedSystem_simple} has
the following two properties 
for every $x(0) \in \mathbb{R}^{\sf{n}}$ and every $\sigma(0) \in \mathcal{P}$:

{\sl Convergence to the origin:~}
$\lim_{t \to \infty}z(t) = 0$.

{\sl Lyapunov stability:~}
To every $\varepsilon > 0$, there corresponds $\delta > 0$ such that
\begin{equation*}
|x(0)| < \delta \quad \Rightarrow \quad
|z(t)| < \varepsilon~~~(t\geq 0).
\end{equation*}
}
\end{theorem}

In the next section,
we shall prove Theorem \ref{thm:stability_theorem}
with the concrete construction of $\mu$.
The sufficient condition on $\tau_a$ is given by \eqref{eq:adt_cond} below.

\section{The proof of Theorem \ref{thm:stability_theorem}}
Let us first consider the fixed ``zoom'' parameter $\mu$.
We obtain the following result on the state trajectories of 
each individual mode.
\begin{lemma}
\label{lem:fix_zoom_parameter}
Fix $p \in \mathcal{P}$, and consider
the non-switched system
\begin{equation}
\label{eq:noSwitchSystem}
\dot z = F_p z + \hat L_p 
(q_{\mu}(y) - y).
\end{equation}
Choose $\kappa > 0$, and
suppose that $M$ satisfies
\begin{equation}
\label{eq:M_Delta_cond}
\sqrt{\underline \lambda_P} M > \sqrt{\overline \lambda_P} \Theta \Delta(1+ \kappa)
C_{\max},
\end{equation}
where $C_{\max}$ and $\Theta$ are defined by 
\eqref{eq:alpha_def} and
\eqref{eq:Theta_def}, respectively.
Then the two ellipsoids
\begin{align*}
\mathscr{R}_1 (\mu,p) &:= 
\left\{
z:~z^{\top}P_pz \leq 
\frac{\underline \lambda_P M^2\mu^2}{C_{\max}^2}
\right\} \\
\mathscr{R}_2 (\mu,p) &:= 
\left\{
z:~z^{\top}P_pz \leq 
\overline \lambda_P (\Theta \Delta (1+\kappa))^2\mu^2
\right\}
\end{align*}
are invariant sets of every trajectory of \eqref{eq:noSwitchSystem}.
Furthermore, if $T$ satisfies
\begin{equation}
\label{eq:T_def}
T >
\frac{\underline \lambda_P M^2 - \overline \lambda_P (\Theta 
\Delta(1+\kappa) C_{\max})^2}
{\underline \lambda_Q \kappa(1+\kappa)(\Theta \Delta C_{\max})^2},
\end{equation}
then every trajectory of \eqref{eq:noSwitchSystem} with an
initial state $z(0) \in \mathscr{R}_1 (\mu)$ satisfies
$z(T) \in \mathscr{R}_2 (\mu)$．
\end{lemma}

\begin{proof}
We prove this lemma in a manner similar to that of Lemma~1 in \cite{Liberzon2003Automatica}.

For every $p \in \mathcal{P}$,
the time derivative of $z^{\top}P_pz$ along the trajectories of
the system \eqref{eq:noSwitchSystem} satisfies
\begin{align}
\frac{d}{dt}(z^{\top}P_pz) 
&= 
-z^{\top} Q_pz +2 z^{\top} P_p \hat L_p
(q_{\mu}(y) - y) \notag \\
&\leq
-\lambda_{\min}(Q_p)|z|^2 + 2 \|P_p \hat L_p\| \cdot |z| \cdot |q_{\mu}(y) - y| \notag \\
&\leq
-\underline \lambda_Q|z|^2 + 2 \max_{p \in P}\|P_p \hat L_p\| 
\cdot |z| \cdot |q_{\mu}(y) - y| \notag \\
&=
-\underline \lambda_Q|z|(|z| - \Theta |q_{\mu}(y) - y|). \label{eq:noSwitch_time_derivative}
\end{align}
On the other hand, 
since $|y| = |C_px| \leq C_{\max}|z|$, it follows from
\eqref{eq:quantizer_cond1_nonSaturation} that,
\begin{equation*}
C_{\max} |z| \leq M\mu
\quad \Rightarrow \quad 
|q_{\mu}(y) - y| \leq \Delta \mu.
\end{equation*}
Hence \eqref{eq:noSwitch_time_derivative} shows
\begin{align}
&\Theta \Delta(1+\kappa)\mu 
\leq |z| \leq \frac{M\mu}{C_{\max}} \notag \\
&\qquad \Rightarrow \quad
\frac{d}{dt}(z^{\top}P_pz) 
\leq
-\underline \lambda_Q \kappa (1+\kappa)(\Theta \Delta \mu)^2.
\label{eq:noSwitch_time_derivative_cond}
\end{align}

If we define the balls $\mathscr{B}_1(\mu)$ and $\mathscr{B_2}(\mu)$ by
\begin{align*}
\mathscr{B}_1(\mu) &:=
\left\{
z:~|z| \leq \frac{M\mu}{C_{\max}}
\right\}\\
\mathscr{B}_2(\mu) &:=
\left\{
z:~|z| \leq \Theta\Delta(1+\kappa)\mu
\right\},
\end{align*}
then it follows from \eqref{eq:alpha_def} and \eqref{eq:M_Delta_cond} that
\begin{equation*}
\mathscr{B}_2(\mu) 
\subset
\mathscr{R}_2(\mu,p)
\subset
\mathscr{R}_1(\mu,p)
\subset
\mathscr{B}_1(\mu) 
\end{equation*}
for $p \in \mathcal{P}$.
Thus \eqref{eq:noSwitch_time_derivative_cond} implies that
$\mathscr{R}_1(\mu,p)$ and $\mathscr{R}_2(\mu,p)$ are invariant sets
of the trajectories of \eqref{eq:noSwitchSystem}.

Also,
the upper bound in
\eqref{eq:noSwitch_time_derivative_cond} shows that
if $x(0)\in \mathscr{R}_1(\mu,p)$, then 
$x(T)\in \mathscr{R}_2(\mu,p)$ for $T$ satisfying \eqref{eq:T_def}.
\end{proof}

We use the next result on average dwell time
for the finite-time estimation of the state at the ``zooming-out'' stage.
Such estimation is needed for Lyapunov stability of the closed-loop system.
\begin{lemma}
\label{lem:ADT_upperbound}
{\em
Fix an initial time $\tau_0 \geq 0$.
Suppose that $\sigma$ satisfies \eqref{eq:ADT_cond}.
Let $\tau \in (0, \tau_a)$, and choose an integer $N$ so that
\begin{equation}
\label{eq:N_ADTcond}
N > \frac{\tau_a}{\tau_a - \tau} \left( N_0 - \frac{\tau}{\tau_a} \right).
\end{equation}
Then there exists a nonnegative real number $T \leq  (N-1)\tau$ such that 
$N_{\sigma}(\tau_0+T+\tau,\tau_0+T) = 0$.
}
\end{lemma}
\begin{proof}
Let us denote the switching times by $t_1,t_2,\dots$, and
fix an integer $N \geq 1$.
Suppose 
\begin{equation}
\label{eq:N_sigma>0}
N_{\sigma}(\tau_0+T+\tau,\tau_0+T) > 0
\end{equation}
for $T \leq (N-1)\tau$. 
Then
we have $t_k - t_{k-1} \leq \tau$ for $k=1,\dots,N$,
where $t_0 := \tau_0$.
Indeed, if 
\begin{equation}
\label{eq:tk_tk-1}
t_k - t_{k-1} > \tau
\end{equation}
for some $k$ and 
if we let $\bar k$ be the smallest integer $k$ satisfying \eqref{eq:tk_tk-1},
then we obtain $t_{\bar k -1} - \tau_0 \leq (\bar k -1)\tau$ and
$N_{\sigma}(t_{\bar k-1}+\tau , t_{\bar k-1}) = 0$, which contradicts
\eqref{eq:N_sigma>0}.
Hence for $0 < \epsilon \leq t_1$,
\begin{equation*}
t_N- (t_1 - \epsilon) = \sum_{k=2}^N (t_k -t_{k-1}) + \epsilon
\leq (N-1)\tau + \epsilon
\end{equation*}
It follows from \eqref{eq:ADT_cond} that
\begin{align*}
N = N_{\sigma}(t_N,t_1 - \epsilon) 
\leq N_0 + \frac{t_N- (t_1 - \epsilon)}{\tau_a} 
\leq N_0 + \frac{(N-1)\tau + \epsilon}{\tau_a}.
\end{align*}
Therefore $N$ satisfies the following inequality:
\begin{equation}
\label{eq:n_upperbound}
N \leq \frac{\tau_a}{\tau_a - \tau} 
\left(
N_0 - \frac{\tau - \epsilon}{\tau_a}
\right).
\end{equation}
Since $\epsilon \in (0, t_1)$ was arbitrary, 
\eqref{eq:n_upperbound} is equivalent to
\begin{equation}
\label{eq:n_upperbound2}
N \leq \frac{\tau_a}{\tau_a - \tau} 
\left(
N_0 - \frac{\tau }{\tau_a}
\right).
\end{equation}
Thus we have shown that
if \eqref{eq:N_sigma>0} holds for all $T \leq (N-1)\tau$,
then $N$ satisfies \eqref{eq:n_upperbound2}.
The contraposition of this statement gives a desired result. 
\end{proof}

\subsection{The proof for convergence to the origin} 
Define $\Gamma$ by
\[\Gamma = \max_{p \in \mathcal{P}}\|A_p\|.\]
We split the proof into two stages:
the ``zooming-out'' and ``zooming-in'' stages.

\subsubsection{\sl The ``Zooming-out'' stage}
Set the control input $u = 0$, and fix $\bar \tau > 0$ and $\chi > 0$.
Then increase $\mu$ in the following way:
$\mu(t)=1$ for $t \in [0,\bar \tau)$, 
$\mu(t)=e^{(1+\chi)k\Gamma \bar \tau}$
for $t \in [k \bar \tau,(k+1)\bar \tau)$ and $k=1,2,\dots$.

Choose $\tau \in (0, \tau_a)$, and
suppose that 
we observe 
\begin{gather}
|q_{\mu(t)}(y(t))| \leq 
M \mu(t) - \Delta \mu(t), \label{eq:zooming_out_QY} \\
\sigma(t) = \sigma(t_0) =: p \label{eq:zooming_out_SS}
\end{gather}
for $t \in [t_0, t_0 + \tau)$.
First we shall describe how to determine $\mu(t_0+\tau)$ 
after this observation,
and next we shall
prove the existence of such $t_0 \geq 0$.

Define the observability Gramian $W_p(\tau)$ by
\begin{equation*}
W_p(\tau) := \int^{\tau}_0 e^{A_p^{\top}t}C_p^{\top} C_pe^{A_pt} dt
\end{equation*}
and the estimated state $\xi(t_0)$ by
\begin{equation}
\label{eq:xi(t0)_def}
\xi(t_0) := W_p(\tau)^{-1}
\int^{\tau}_{0}  e^{A_p^{\top}t}C_p^{\top} q_{\mu(t_0+t)}(y(t_0+t))dt
\end{equation}
Since $u(t) = 0$, we also have
\begin{equation}
\label{eq:x(t0)_for_xi(t0)}
x(t_0) = W_p(\tau)^{-1}
\int^{\tau}_{0}  e^{A_p^{\top}t}C_p^{\top} y(t_0+t)dt.
\end{equation}
Moreover,
if \eqref{eq:zooming_out_QY} holds, then
\eqref{eq:quantizer_cond2} gives
\begin{equation*}
\left| y(t) \right| 
\leq 
M\mu(t)
\qquad (t_0 \leq t < t_0 + \tau),
\end{equation*}
and hence
\begin{equation*}
|q_{\mu(t)}(y(t)) - y(t)| \leq \Delta \mu(t) 
\qquad (t_0 \leq t < t_0 + \tau).
\end{equation*}
Therefore \eqref{eq:xi(t0)_def} and \eqref{eq:x(t0)_for_xi(t0)} show that
\begin{equation*}
|x(t_0) - \xi(t_0)| \leq \|W_p(\tau)^{-1}\|\tau \Lambda_p^C(\tau) 
\Delta \mu ^- (t_0+\tau),
\end{equation*}
where
\begin{align*}
\Lambda_p^C(\tau) &:= \max_{0\leq t \leq \tau} \left\| C_pe^{A_pt} \right\| \\
\mu ^- (t_0+\tau) &:= \lim_{t \nearrow t_0 + \tau}\mu(t).
\end{align*}
Since $x(t_0+ \tau) = e^{A_p \tau} x(t_0)$, if we set 
\begin{equation}
\label{eq:xi(t0+tau)_def}
\xi(t_0+ \tau) = e^{A_p\tau} \xi(t_0),
\end{equation}
then 
\begin{align*}
&|x(t_0+ \tau) - \xi(t_0+ \tau)| \\
&~ \leq 
\|W_p(\tau)^{-1}\|\tau \Lambda_p^C (\tau) \left\| e^{A_p\tau} \right\| 
\Delta \mu ^- (t_0+\tau)
=: e(t_0+\tau).
\end{align*}

It follows that
\begin{align}
|z(t_0 + \tau)| 
&\leq
|x(t_0 + \tau)| + |x(t_0 + \tau) - \xi(t_0 + \tau)| \notag \\
&\leq
|\xi(t_0 + \tau)| + 2|x(t_0 + \tau) - \xi(t_0 + \tau)| \notag \\
&\leq
|\xi(t_0 + \tau)| + 2 e(t_0+\tau)
=: E(t_0+\tau).
\notag
\end{align}
Thus if we choose $\mu(t_0+\tau)$ so that
\begin{equation}
\label{eq:z(ta0+Ntau)2}
\mu(t_0 + \tau) \geq
\sqrt{\frac{\overline \lambda_P}{\underline \lambda_P}}
\frac{C_{\max}}{M} E(t_0+\tau),
\end{equation}
then $z(t_0 + \tau)\in 
\mathscr{R}_1(\mu(t_0 + \tau),\sigma(t_0 + \tau))$.

It remains to prove the existence of $t_0 \geq 0$ satisfying 
\eqref{eq:zooming_out_QY} and \eqref{eq:zooming_out_SS} for 
$t \in [t_0,t_0 + \tau)$.
By the definition of $\mu$ and \eqref{eq:M_Delta_cond_main_thm}, 
there is $\tau_0 \geq 0$ such that
\begin{equation*}
\left|
y(t)
\right| \leq M\mu(t) -2\Delta \mu(t)
\qquad (t \geq \tau_0).
\end{equation*}
In conjunction with \eqref{eq:quantizer_cond1_nonSaturation},
this implies that \eqref{eq:zooming_out_QY} holds for $t \geq \tau_0$.
Let $N$ be an integer satisfying \eqref{eq:N_ADTcond}.
Then Lemma \ref{lem:ADT_upperbound} guarantees the
existence of $t_0 \in [\tau_0, \tau_0 + (N-1)\tau]$ such that 
\eqref{eq:zooming_out_SS} holds for $t \in [t_0, t_0 + \tau)$.

\subsubsection{\sl The ``Zooming-in'' stage}
Choose $\kappa$ so that \eqref{eq:M_Delta_cond} holds,
and define $T_0 := t_0+\tau$.
We consider \eqref{eq:ClosedSystem} with $\xi(T_0)$ 
calculated by \eqref{eq:xi(t0)_def} and \eqref{eq:xi(t0+tau)_def}.
The discussion above ensures 
$z(T_0)\in \mathscr{R}_1(\mu(T_0),\sigma(T_0))$.
Fix $T$ so that \eqref{eq:T_def} is satisfied.

Let us first investigate the case 
without switching on the interval $(T_0, T_0+T]$.
In this case, 
if we let $\mu(t) = \mu(T_0)$ for $t \in [T_0, T_0+T)$, then
Lemma \ref{lem:fix_zoom_parameter} shows that
$z(T_0 + T)\in \mathscr{R}_2(\mu(T_0),\sigma(T_0))$.
Define $\Omega$ by
\begin{equation}
\label{eq:Omega_def}
\Omega := \sqrt{\frac{\overline \lambda_P}{\underline \lambda_P}}
\frac{\Theta\Delta (1+\kappa)C_{\max}}{M},
\end{equation}
and set $\mu(T_0+T) = \Omega \mu(T_0)$.
Then we obtain $z(T_0+T) \in \mathscr{R}_1(\mu(T_0+T),\sigma(T_0+T))$.
Note that $\Omega < 1$ by \eqref{eq:M_Delta_cond_main_thm}.
As regards after $T_0+T$,
if switching does not occur on the interval $(T_0+mT, T_0+(m+1)T]$ for $m=1,2\dots$, 
then we update $\mu$ in the same way.

We now study the switched case. 
Let $T_1,T_2,\dots,T_n$ be switching times on the interval $(T_0, T_0+T]$.
We sometimes write $T_{n+1}$ rather than $T_0+T$ for simplicity of notation.
Suppose that
for every $p_1,p_2 \in \mathcal{P}$ with $p_1 \not= p_2$,
there exists $c_{p_2,p_1} > 0$ such that
\begin{equation}
\label{eq:Lyapunov_cond}
z^{\top} P_{p_1} z \leq  c_{p_2,p_1} \cdot z^{\top} P_{p_2} z.
\end{equation}
for all $z \in \mathbb{R}^{2{\sf n}}$.
We adjust $\mu$ at every switching time in the following way:
\begin{align*}
\mu(t) =
\sqrt{
\prod _{\ell = 0}^{k-1}c_{\sigma(T_{\ell + 1}), \sigma(T_{\ell})}
} \cdot \mu(T_0) \qquad (T_k \leq t < T_{k+1})
\end{align*}
for $k=0,\dots,n$.

Lemma \ref{lem:fix_zoom_parameter} suggests that
$\mathscr{R}_i (\mu(T_k),\sigma(T_k))$ ($i=1,2$)
are invariant sets for $t \in [T_k,T_{k+1})$, $k=0,\dots,n$.
Moreover, by \eqref{eq:Lyapunov_cond},
if $z(t) \in \mathscr{R}_i (\mu_0,p_1)$ for some $\mu_0$, then
$z(t) \in \mathscr{R}_i (\sqrt{c_{p_2,p_1}}\mu_0,p_2)$ ($i=1,2$) for
$p_1,p_2 \in \mathcal{P}$ with $p_1\not=p_2$.
Hence it follows that $z(t) \in \mathscr{R}_1 (\mu(t),\sigma(t))$
for $t \in [T_0,T_{n+1})$.
Also, if there is $t_1 \in [T_0,T_{n+1})$ such that
$x(t_1) \in \mathscr{R}_2 (\mu(t_1),\sigma(t_1))$, then
$z(t) \in \mathscr{R}_2 (\mu(t),\sigma(t))$ for all $t \in [t_1, T_{n+1})$.
To see the existence of such $t_1$, suppose for a contradiction that 
\begin{equation}
\label{eq:z_NOT_R2}
z(t) \not\in \mathscr{R}_2 (\mu(t),\sigma(t)),\qquad (T_0 \leq t < T_{n+1}).
\end{equation}

First we examine the case $T_{n+1} > T_n$.
Since a Filippov solution is (absolutely) continuous, it follows from
\eqref{eq:z_NOT_R2} that
\begin{equation}
\label{eq:Lyapunov_T_0+T}
\lim_{t \nearrow T_{n+1}} z^{\top}P_{\sigma(t)}z 
\geq \overline \lambda_P (\Theta \Delta (1+\kappa))^2\mu(T_n)^2.
\end{equation}

On the other hand, since 
$z(t) \in \mathscr{B}_1(\mu(t))$ and $z(t) \not\in \mathscr{B}_2(\mu(t))$ for
$t < T_{n+1}$, 
\eqref{eq:noSwitch_time_derivative_cond} shows that
\begin{align*}
\lim_{t \nearrow T_1}z^{\top}P_{\sigma(t)}z 
\leq
\frac{\underline \lambda_P M^2\mu(T_0)^2}{C_{\max}^2}
-
(T_1-T_0)\underline \lambda_Q \kappa(1+\kappa)(\Theta \Delta \mu(T_0))^2,
\end{align*}
and hence we have
\begin{align*}
z^{\top}P_{\sigma(T_1)}z 
&\leq
c_{\sigma(T_1), \sigma(T_0)} 
\cdot \left( \lim_{t \nearrow T_1}z^{\top}P_{\sigma(t)}z\right) \\
&=
\left(
\frac{\underline \lambda_P M^2}{C_{\max}^2}
-
(T_1-T_0)\underline \lambda_Q \kappa(1+\kappa)(\Theta \Delta)^2
\right)\mu(T_1)^2.
\end{align*}

If we repeat this process and use \eqref{eq:T_def},
then
\begin{align}
\lim_{t \nearrow T_{n+1}}
z^{\top}P_{\sigma(t)}z 
\leq
\left(
\frac{\underline \lambda_P M^2}{C_{\max}^2}
-
T \underline \lambda_Q \kappa (1+\kappa)(\Theta \Delta )^2
\right) \mu(T_n)^2 
<
\overline \lambda_P (\Theta \Delta (1+\kappa))^2\mu(T_n)^2,
\label{eq:Lyapunov_T_0+T_contradiction}
\end{align}
which contradicts \eqref{eq:Lyapunov_T_0+T}.
Hence we obtain
\begin{equation}
\label{eq:Lyapunov_T_0+T_result}
z(T_{n+1}) = 
\lim_{t \nearrow T_0+T} z(t)
\in \mathscr{R}_2 (\mu(T_n),\sigma(T_n)).
\end{equation}

In the case $T_{n+1} = T_n$,
\eqref{eq:Lyapunov_T_0+T} and \eqref{eq:Lyapunov_T_0+T_contradiction} hold
with $T_{n-1}$ in place of $T_n$, and then we have
\begin{align}
z(T_{n+1}) &= 
\lim_{t \nearrow T_{n+1}} z(t) \notag \\
&\in \mathscr{R}_2 (\mu(T_{n-1}),\sigma(T_{n-1})) \notag \\
&\subset \mathscr{R}_2 (\sqrt{c_{\sigma(T_n), \sigma(T_{n-1})}} 
\cdot \mu(T_{n-1}),\sigma(T_n)). \notag
\end{align}

Thus if $n$ switches occur, then we set 
\[
\mu(T_0+T) = \Omega 
\sqrt{\prod _{\ell = 0}^{n-1}c_{\sigma(T_{\ell + 1}), \sigma(T_{\ell})}} 
\cdot \mu(T_0).
\]
The discussion above implies
$z(T_0+T) \in \mathscr{R}_1(\mu(T_0+T),\sigma(T_0+T))$.
We update $\mu$  in the same way after $T_0+T$.

Finally, 
define 
\begin{equation}
\label{eq:c_def_cont}
c := \max_{p_1 \not= p_2} c_{p_2,p_1}.
\end{equation}
Then
\eqref{eq:ADT_cond} gives
\begin{align}
\mu(T_0+mT) 
\leq \Omega^m \sqrt{c^{N_{\sigma}(T_0+mT, T_0)}}\mu(T_0) 
\leq \sqrt{c^{N_0}} \cdot
	\left(\Omega \sqrt{c^{T/\tau_a}} \right)^m \mu(T_0)
\label{eq:mu(T_0+MT)}
\end{align}
for $m \in \mathbb{N}$.
If $\Omega \sqrt{\lambda^{T/\tau_a}} < 1$, that is, if
the average dwell time $\tau_a$ satisfies
\begin{equation}
\label{eq:adt_cond}
\tau_a > \frac{\log(c)}{2\log(1/\Omega)} T,
\end{equation}
then $\lim_{m\to \infty}\mu(T_0+mT) =0$.
Since $x(t) \in \mathscr{B}_1(\mu(t))$ for $t \geq T_0$, we obtain
$\lim_{t \to \infty} x(t)= 0$.
\hfill $\blacksquare$

\begin{remark}
\noindent
{(a)~}
The proposed method of adjusting $\mu$ is causal but 
sensitive to the time-delay of the switching signal at the ``zooming-in'' stage.
To allow such a delay, we must examine the bound of an error due to
the mismatch of modes between the plant and the controller.
However we do not proceed along this line to avoid technical issues.

\noindent
{(b)~}
Here we have changed $\mu$ at every switching time in the ``zooming-in'' stage.
If we would not, switching might lead to instability of the closed-loop system.
Without adjustment of $\mu$, 
the quantizer does not saturate right after the switch, because
the trajectory belongs to $\mathscr{B}_1(\mu)$.
However, $\mathscr{B}_1(\mu)$ is not an invariant set, so
if we do not change $\mu$, the trajectory may leave $\mathscr{B}_1(\mu)$. 
This leads to saturation of the quantizer.


\end{remark}

\subsection{The proof for Lyapunov stability} 
The proof of Lyapunov stability follows in a line similar to that 
in Sec.~5.5 of \cite{Liberzon2013}.

Let us denote by $\mathscr{B}_{\varepsilon}$ 
the open ball with center at the origin and radius $\varepsilon$
in $\mathbb{R}^{2{\sf n} \times 2 {\sf n}}$.
In what follows,
we use the letters in the previous subsection
and assume that \eqref{eq:adt_cond} holds.

Let $\delta > 0$ be small enough to satisfy
\begin{equation}
\label{eq:delta_cond_for_u=0}
C_{\max} e^{\Gamma N\tau} \delta < \Delta_0.
\end{equation}
Then $q_{\mu(t)}(y(t)) = 0$ for $t \in [0, N\tau]$.
The argument on the existence of $t_0$ at the ``zooming-out'' stage implies that
the time $T_0$, at which the stage changes from ``zooming-out'' to ``zooming-in'',
satisfies $T_0 \leq N\tau$ for every switching signal.

Fix $\alpha > 0$.
By \eqref{eq:xi(t0)_def},
$\xi(T_0) = 0$, and hence we see from \eqref{eq:z(ta0+Ntau)2} 
that $\mu(T_0)$ achieving 
$z(T_0) \in \mathscr{R}_1(\mu(T_0),\sigma(T_0))$
can be chosen so that
\begin{equation}
\label{eq:bar_mu_ineq}
\alpha \leq
\mu(T_0) \leq
\bar \mu,
\end{equation}
where $\bar \mu$ is defined by
\begin{align*}
\bar \mu &:=
\max
\Biggl\{\alpha,~~
\sqrt{\frac{\overline \lambda_P}{\underline \lambda_P}}
\frac{ \Delta \tau C_{\max}
e^{(1+\chi)\lfloor N\tau / \bar \tau \rfloor \Gamma \bar \tau}}{M}
\times
\max_{p \in \mathcal{P}} 
\left(
\|W_p(\tau)^{-1}\|\Lambda_p^C(\tau)  \left\| e^{A_p\tau} \right\|
\right)
\Biggr\}.
\end{align*}
Note that $\bar \mu$ is independent on switching signals.

By \eqref{eq:mu(T_0+MT)}, if $m$ satisfies 
\begin{equation}
\label{eq:m_cond_Lyapunov}
m > \frac{\log(\bar \mu M \sqrt{c^{N_0}}/(\varepsilon C_{\max}))}
{\log (1/(\Omega\sqrt{c^{T/\tau_a}}))},
\end{equation}
then we have
\begin{equation}
\label{eq:R_1(T1)}
\mathscr{R}_1(\mu(T_0+mT),\sigma(T_0+mT)) \subset \mathscr{B}_{\varepsilon}.
\end{equation}
Let $\bar m$ be the smallest integer satisfying \eqref{eq:m_cond_Lyapunov}.

Define
$T_1 := T_0+\bar m T \leq N\tau + \bar m T$
and 
\[
\underline{c} := \min_{p_1 \not= p_2}c_{p_2,p_1}.
\]
By \eqref{eq:bar_mu_ineq}, we have
\begin{align}
\mu(t) 
&\geq \Omega^{\bar m}
\sqrt{\underline{c}^{N_0+\bar m T/\tau_a}}\mu(T_0) \notag \\
&\geq 
\alpha\Omega^{\bar m}
\sqrt{\underline{c}^{N_0+\bar m T/\tau_a}}=: \eta
\label{eq:mu_Omega_bound}.
\end{align}
for $t \in [T_0, T_1]$.
Let
$\delta > 0$ satisfy
\begin{gather}
\label{eq:delta_cond1}
C_{\max} e^{\Gamma (N\tau + \bar m T)} \delta < 
\eta \Delta_0\\
\label{eq:delta_cond2}
e^{\Gamma (N\tau + \bar m T)} \delta < 
\min \left\{ \varepsilon,~~
\sqrt{\frac{\underline \lambda_P}{\overline \lambda_P} } 
\frac{M \eta}{C_{\max}} \right\}.
\end{gather}

By \eqref{eq:mu_Omega_bound} and \eqref{eq:delta_cond1}, 
$q_{\mu(t)}(y(t))=0$ on the interval $[0,T_1]$, so
$\xi(t) = 0$ and $u(t) = 0$ on the same interval.
Combining this with \eqref{eq:delta_cond2}, we obtain
$|x(t)| \leq e^{\Gamma (N\tau + \bar m T)} \delta < \varepsilon$
for $t \leq T_1$. 
Thus 
\begin{equation}
\label{eq:z_eps1}
|z(t)| \leq |x(t)| + |\xi(t)|< \varepsilon \qquad (t \leq T_1).
\end{equation}

On the other hand, by \eqref{eq:mu_Omega_bound} and \eqref{eq:delta_cond2},
\begin{align*}
z(T_1)^{\top} P_{\sigma(T_1)} z(T_1) &\leq
\lambda_{\max}(P_{\sigma(T_1)}) |z(T_1)|^2 \\
&<
\frac{\underline \lambda_P M^2\eta^2}{C_{\max}^2}
\leq
\frac{\underline \lambda_P M^2\mu(T_1)^2}{C_{\max}^2}
\end{align*}
for every $p \in \mathcal{P}$, and
hence $z(T_1) 
\in \mathscr{R}_1(\mu(T_1),\sigma(T_1))\subset \mathscr{B}_{\varepsilon}$ by
\eqref{eq:R_1(T1)}.
In addition, since
\begin{align*}
\mu(T_1+kT) &\leq 
\sqrt{c^{N_0}} \cdot
	\left(\Omega \sqrt{c^{T/\tau_a}} \right)^{\bar m+k} \mu(T_0) \\
&\leq
\sqrt{c^{N_0}} \cdot
	\left(\Omega \sqrt{c^{T/\tau_a}} \right)^{\bar m} \bar \mu
\end{align*}
for all $k \geq 0$ and since $\bar m$ satisfies
\eqref{eq:m_cond_Lyapunov}, it follows that that 
$\mathscr{R}_1(\mu(T_1+kT),\sigma(T_1+kT))$ also lies in
$\mathscr{B}_{\varepsilon}$.
Since $\mathscr{R}_1(\mu(t),\sigma(t))$ is
an invariant set for $t \geq T_0$, we have
\begin{equation}
\label{eq:z_eps2}
|z(t)| < \varepsilon \qquad (t \geq T_1).
\end{equation}
From \eqref{eq:z_eps1} and \eqref{eq:z_eps2}, we see that 
Lyapunov stability can be achieved.
\hfill $\blacksquare$

\begin{remark}
Through 
Lemma \ref{lem:ADT_upperbound},
we implicitly use the average dwell time property to obtain
the upper bound $\bar \mu$ in \eqref{eq:bar_mu_ineq}.
\end{remark}

\section{Numerical Examples}
Consider the continuous-time switched system \eqref{eq:ClosedSystem} with
the following two modes:
\begin{gather*}
A_1 = 
\begin{bmatrix}
1 & 0 \\ 0 & -4
\end{bmatrix},\quad
B_1 = 
\begin{bmatrix}
1 \\ 0 
\end{bmatrix},\quad
C_1 = 
\begin{bmatrix}
1 & 1
\end{bmatrix},\\
A_2 = 
\begin{bmatrix}
0 & 1 \\ -1 & 0
\end{bmatrix},\quad
B_2 = 
\begin{bmatrix}
0 \\ 1
\end{bmatrix},\quad
C_2 = 
\begin{bmatrix}
0 & -1
\end{bmatrix}.
\end{gather*}
As the feedback gain and the observar gain of each mode, we take
\begin{gather*}
K_1 = 
\begin{bmatrix}
-3 & 1
\end{bmatrix},~
L_1 = 
\begin{bmatrix}
-2 \\ 0 
\end{bmatrix},~
K_2 = 
\begin{bmatrix}
0 & -1
\end{bmatrix},~
L_2 = 
\begin{bmatrix}
0 \\ 1
\end{bmatrix}.
\end{gather*}
Let the quantizer $q$ be uniform-type, and
define the parameters $M$ and $\Delta$ of the quantizer by
$
M = 20$,
$
\Delta = 0.1,
$
Also, define $Q_1$ and $Q_2$ in \eqref{eq:Lyapunov_cont} and $\kappa$ in 
\eqref{eq:M_Delta_cond} by
$
Q_1 = 
\diag(2,8,2,8)$,
$
Q_2 =
\diag(1,1,1,1),
$
$\kappa = 2.5$,
where $\diag(e_1,\dots,e_4)$ means a diagonal matrix whose diagonal
elements starting in the upper left corner are $e_1,\dots,e_4$.
Then we obtain 
$T \approx 2.20$ in \eqref{eq:T_def},
$\Omega \approx 0.824$ in \eqref{eq:Omega_def}, 
$c \approx 4.03$ in \eqref{eq:c_def_cont}, and
$\tau_a \approx 7.90$ in \eqref{eq:adt_cond}.

Fig.~\ref{fig:cont_simulation} (a) and (b) show that the output $y$ and
the $\ell^2$-norm of the state $x$ of the switched system~\eqref{eq:SLS}
with $x(0) = [-6~~5]^{\top}$ and $\mu(0) = 1$.
In this example, the ``zooming-out'' stage finished at $t = 0.5$.
We see the non-smooth behaviors of $y$ and $x$ at the switching times $t=5,20,28,36$.
In particular, we observe from the behaviors of $x$ at $t=5,28$ that, not surprisingly,
adjustments of $\mu$ at some swithcing times are conservative.

\begin{figure}[t]
\centering
\subcaptionbox{The output $y$.}
{\includegraphics[width = 7.2cm,bb= 15 10 700 530,clip]{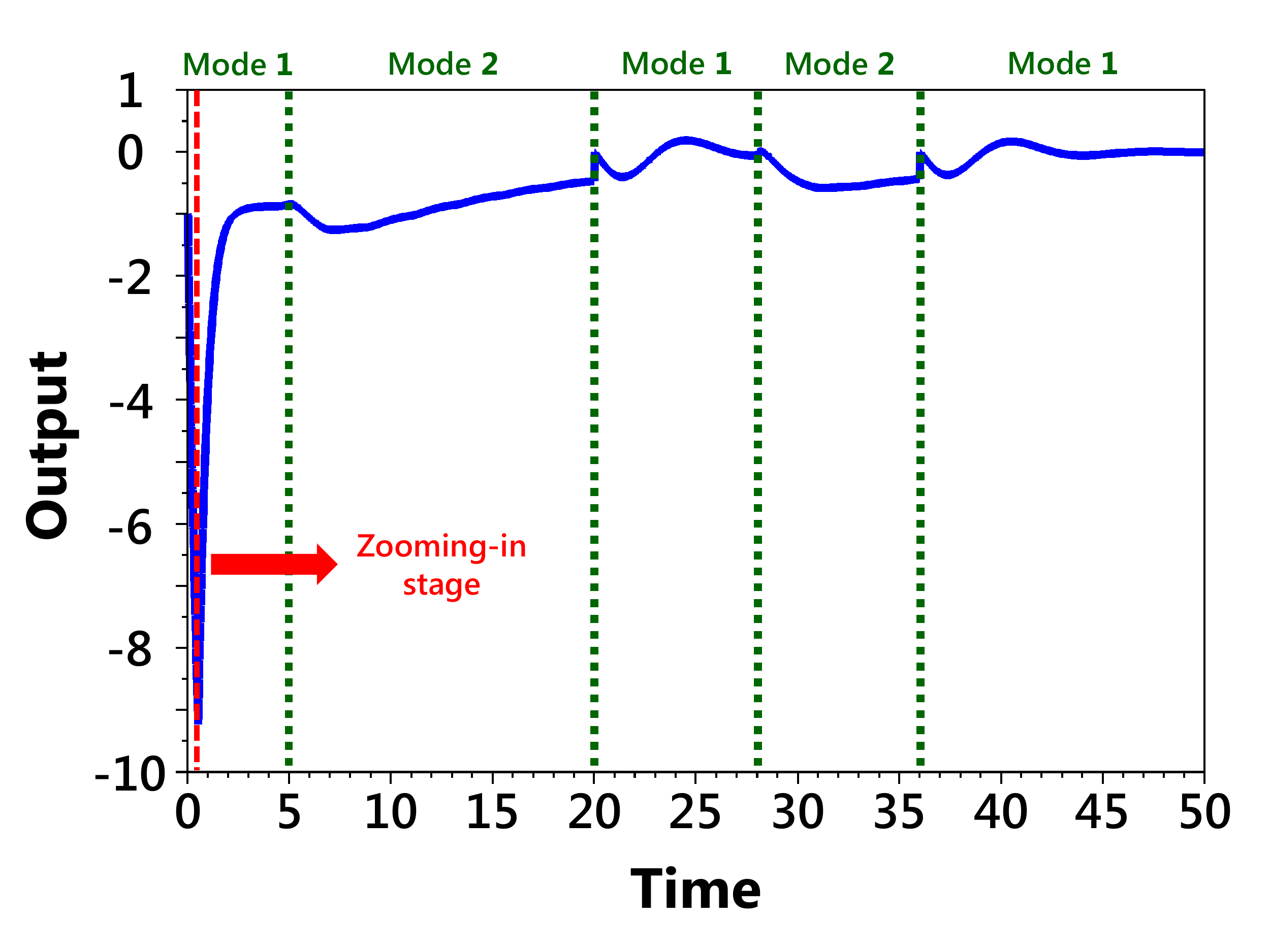}}

\subcaptionbox{The $\ell^2$-norm of the state $x$.}
{\includegraphics[width = 6.8cm,bb= 25 10 690 550,clip]{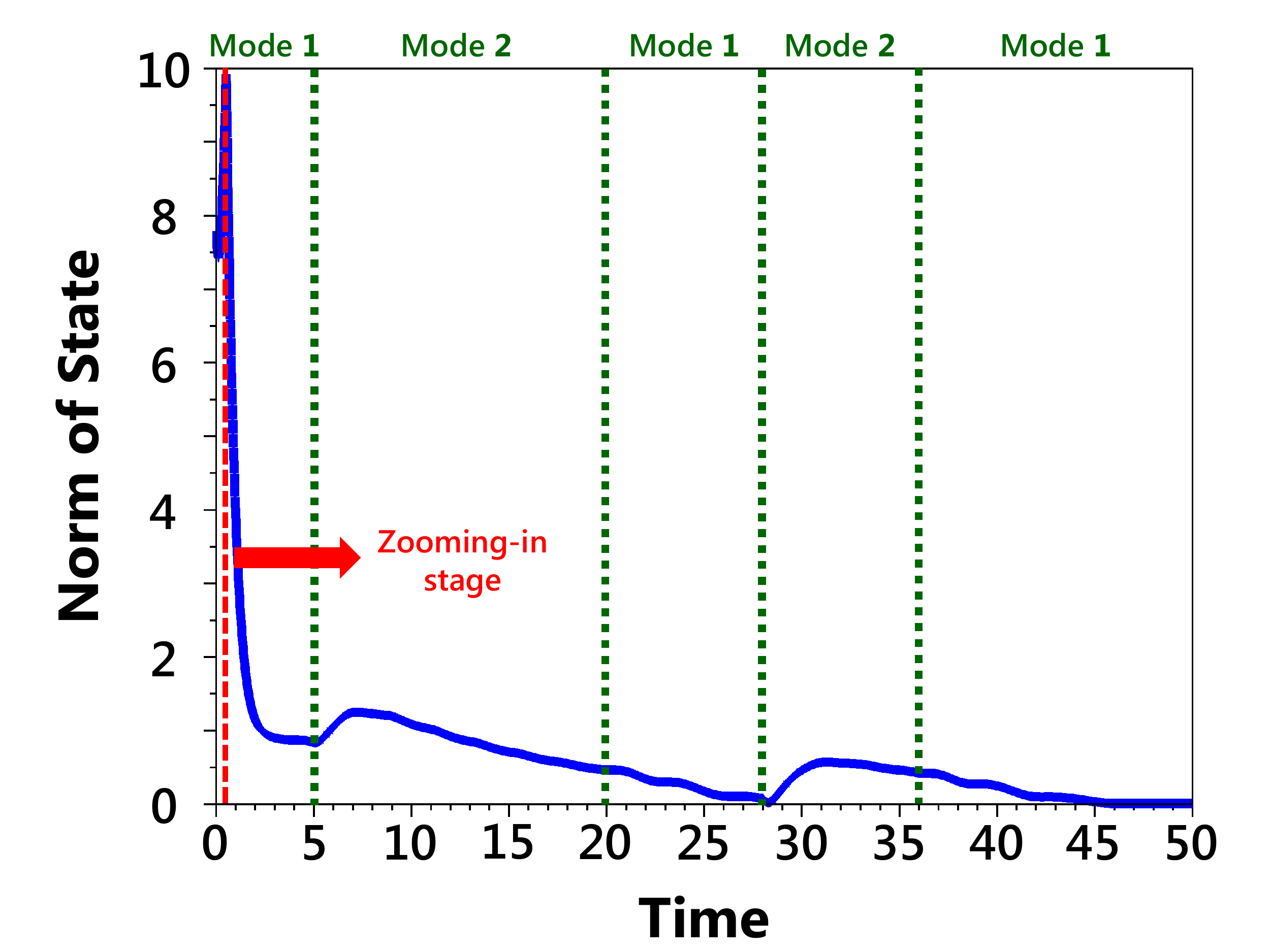}}
\caption{Simulation with $x(0)= [-6~5]^{\top}$ and $\sigma(0) = 1$.
\label{fig:cont_simulation}}
\end{figure}


\section{CONCLUDING REMARKS}
The stabilization of continuous-time switched linear systems
by quantized output feedback has been studied.
We have proposed an output encoding method for
globally asymptotic stability.
The encoding method is rooted in the non-switched case, and
an additional adjustment of the zoom parameter 
is needed at every switching time
in the zooming-in stage.
We have discussed the effect of switching by
using multiple Lyapunov functions and an average dwell time assumption.



\bibliographystyle{IEEEtran}

\end{document}